\documentclass[12pt]{article}
\usepackage{amsmath}
\usepackage{amsthm}
\usepackage{amssymb}
\usepackage{txfonts}
\usepackage{mathrsfs}
\usepackage{arydshln}
\usepackage{blkarray}
\usepackage{graphicx}
\usepackage{cite}
\usepackage{bm}
\usepackage{color}
\usepackage{xcolor}
\usepackage{polynom}
\usepackage{framed}
\usepackage{tikz}
\usepackage{multirow}
\usetikzlibrary{decorations.pathreplacing,calligraphy}
\usepackage{hyperref}
\usepackage{multirow}  
\usepackage{booktabs}
\usepackage{bm}
\usepackage{enumitem}
\definecolor{red}{rgb}{1,0,0}
\definecolor{blue}{rgb}{0,0,1}
\definecolor{darkblue}{rgb}{0,0,0.4}

\begin{document}
\newtheorem{Def}{Definition}[section]
\newtheorem{exe}{Example}[section]
\newtheorem{prop}[Def]{Proposition}
\newtheorem{theo}[Def]{Theorem}
\newtheorem{lem}[Def]{Lemma}
\newtheorem{rem}{\noindent\mbox{Remark}}[section]
\newtheorem{coro}[Def]{Corollary}
\newcommand{\Ker}{\rm Ker}
\newcommand{\Soc}{\rm Soc}
\newcommand{\lra}{\longrightarrow}
\newcommand{\ra}{\rightarrow}
\newcommand{\add}{{\rm add\, }}
\newcommand{\gd}{{\rm gl.dim\, }}
\newcommand{\End}{{\rm End\, }}
\newcommand{\overpr}{$\hfill\square$}
\newcommand{\rad}{{\rm rad\,}}
\newcommand{\soc}{{\rm soc\,}}
\renewcommand{\top}{{\rm top\,}}
\newcommand{\fdim}{{\rm fin.dim}\,}
\newcommand{\fidim}{{\rm fin.inj.dim}\,}
\newcommand{\gldim}{{\rm gl.dim}\,}
\newcommand{\cpx}[1]{#1^{\bullet}}
\newcommand{\D}[1]{{\mathscr D}(#1)}
\newcommand{\Dz}[1]{{\rm D}^+(#1)}
\newcommand{\Df}[1]{{\rm D}^-(#1)}
\newcommand{\Db}[1]{{\mathscr D}^b(#1)}
\newcommand{\C}[1]{{\mathscr C}(#1)}
\newcommand{\Cz}[1]{{\rm C}^+(#1)}
\newcommand{\Cf}[1]{{\rm C}^-(#1)}
\newcommand{\Cb}[1]{{\mathscr C}^b(#1)}
\newcommand{\K}[1]{{\mathscr K}(#1)}
\newcommand{\Kz}[1]{{\rm K}^+(#1)}
\newcommand{\Kf}[1]{{\rm K}^-(#1)}
\newcommand{\Kb}[1]{{\mathscr K}^b(#1)}
\newcommand{\modcat}[1]{#1\mbox{{\rm -mod}}}
\newcommand{\stmodcat}[1]{#1\mbox{{\rm -{\underline{mod}}}}}
\newcommand{\pmodcat}[1]{#1\mbox{{\rm -proj}}}
\newcommand{\imodcat}[1]{#1\mbox{{\rm -inj}}}
\newcommand{\opp}{^{\rm op}}
\newcommand{\otimesL}{\otimes^{\rm\bf L}}
\newcommand{\rHom}{{\rm\bf R}{\rm Hom}\,}
\newcommand{\pd}{{\rm pd}}
\newcommand{\Hom}{{\rm Hom \, }}
\newcommand{\Coker}{{\rm coker}\,\,}
\newcommand{\Ext}{{\rm Ext}}
\newcommand{\im}{{\rm im}}
\newcommand{\Cone}{{\rm Cone}}
\newcommand{\fini}{{finitistic}}
\newcommand{\proj}{{\rm  proj}}
\newcommand{\al}{{\alpha}}
\newcommand{\Tr}{{\rm trace}}
\newcommand{\tr}{{\rm tr}}
\newcommand{\Tor}{{\rm  Tor}}
\newcommand{\ad}{{\rm  add}}
\newcommand{\StHom}{{\rm \underline{Hom} \, }}
\newcommand{\F}{{\mathbb{F}_q}}

\renewcommand{\theequation}{\arabic{section}.\arabic{equation}}
\pagenumbering{arabic}{\Large \bf

\begin{center}
	Multilevel constructions of constant dimension codes based on one-factorization  of  complete graphs
\end{center}}
\centerline{ {\sc Dengming Xu$^1$\qquad Mengmeng Li$^2$   }}
\begin{abstract}
Constant dimension codes  (CDCs) have become an important object  in  coding theory due to their application in random network coding.  The multilevel construction is one of the most effective ways to  construct constant dimension codes. The paper is devoted to constructing  CDCs by  the multilevel construction. Precisely,   we first choose  an appropriate  skeleton code   based on the transformations of binary vectors related to
the one-factorization of complete graphs;  then we construct CDCs by using the  chosen skeleton code, where quasi-pending blocks are used; finally, we calculate the dimensions
by use of  known constructions
of optimal Ferrers diagram rank metric codes.
As applications, we improve the  lower bounds of  $\overline{A}_q(n,8,6)$ for $16\leq n\leq 19.$
\end{abstract}

\vskip-1cm

\renewcommand{\thefootnote}{\alph{footnote}}
\setcounter{footnote}{-1} \footnote{$^1$ Corresponding author, Sino-European Institute of Aviation Engineering,
Civil Aviation University of China, 300300 Tianjin,
PR China. E-mail:
xudeng17@163.com.
}
\setcounter{footnote}{-1} \footnote{$^2$ College of Science,
Civil Aviation University of China, 300300 Tianjin,
PR  China. E-mail: limengm00@163.com
}

\section{Introduction}
Let $\mathbb{F}_{q}$ be the finite field of order $q$ and $\mathbb{F}_{q}^{n}$ be the $n$ dimensional vector space over $\mathbb{F}_{q}$. Let $\mathcal{P}_{q}(n)$ be the set of all subspaces of $\mathbb{F}_{q}^{n}$. Given a nonnegative integer $ k \leq n$, denote  by  $\mathcal{G}_{q}(n,k)$ the set of all $k$-dimensional subspaces of $\mathbb{F}_{q}^{n}$. It is well-known  that $$	
\left| \mathcal{G}_q(n, k) \right|=\left[\begin{smallmatrix}
n\\k
\end{smallmatrix}\right]_q:=\prod_{i=1}^{k} \frac{q^{n-k+i}-1}{q^i-1}
$$ where $\left[\begin{smallmatrix}
n\\k
\end{smallmatrix}\right]_q$ is the Gaussian  coefficient.
For any two subspaces $U,V\in \mathcal{P}_{q}(n)$, the subspace distance  between $U$ and $V$ is defined by
\begin{center}
$d_{S}(U,V)=\dim(U)+\dim(V)-2\dim(U\cap V).$
\end{center}
A subset $\mathscr{C}$ of $\mathcal{G}_{q}(n,k)$ is called an $\left( n,M,d,k\right) _{q}$ constant dimension code (CDC) if it has cardinality $M$ and $d=\min\{d_{S}(U,V) \mid (U,V)\in \mathscr  C^2, U\neq V\}.$ Denote by $A_{q}(n,d,k)$   the maximum size of an $(n,M,d,k)_{q}$ code. Then an $(n,M,d,k)_{q}$-CDC  is said to be optimal if $M=A_{q}(n,d,k)$.

Constant dimension codes have received significant attention since Koetter and Kschischang  presented an application of subspaces codes for error correction in random network coding  \cite{koetter2008coding}. One main problem on CDCs is to improve the lower bound of $A_{q}(n,d,k)$
for given parameters $n,d,k.$ In \cite{silva2008rank}, it was proved that  a simple construction based on lifting of maximum rank distance (MRD) codes can produce asymptotically optimal CDCs. However, there are still big gaps between the sizes of lifted MRD codes and known upper bounds for CDCs.
Since 2008, various methods have been used to construct CDCs for improving the lower bound of $A_q(n,d,k)$, such as multilevel construction\cite{etzion2009error,liu2020parallel,trautmann2010new,silberstein2015subspace}, linkage construction\cite{gluesing2016construction,heinlein2019new,huimin2022new}, coset construction\cite{heinlein2017coset}, and so on.

Multilevel construction  was first introduced  by Etzion and Silberstein in \cite{etzion2009error}, which is a  generalization of  the construction of the lifted MRD codes.
In \cite{trautmann2010new}, Trautmann and Rosenthal indicated that the pending dots  of an Ferrers diagram can increase  the subspace distance.  In \cite{etzion2012codes}, Etzion and Silberstein  improved the multilevel construction by using pending dots and complete graphs.
In \cite{silberstein2015subspace} , Silberstein and Trautmann further generalized multilevel construction by using  quasi-pending blocks.  Multilevel construction has become an effective way and was widely used to  construct CDCs, see also \cite{liu2023constant,liu2022double}.

In the paper, we focus on constructing $(nt,4t,3t)_q$ CDCs by the multilevel construction.  Precisely,   we first choose  an appropriate  skeleton code  by   the extensions and permutations of binary vectors related to  one-factorization of complete graphs, this provides a new idea to choose the skeleton codes in the multilevel construction;  then, inspired by the work in \cite{etzion2012codes,silberstein2015subspace},  we construct CDCs by using the chosen  skeleton code, where quasi-pending blocks are used; finally, we calculate the dimensions
by use of the known constructions
of optimal Ferrers diagram rank metric codes.
As applications, we improve the  lower bounds of  $\overline{A}_q(n,8,6)$ for $16\leq n\leq 19.$

The paper is organized as follows.	In Section 2, we  introduce some definitions and known results which will be used in the paper. In Section 3, we give the constructions of $(nt,4t,3t)_q$  CDCs and calculate the dimensions, as applications of the main result, new lower bounds of $\overline{A}_q(n,8,6)$  for $16\leq n\leq 19$ are provided . In  Section 4, we give the conclusions of the  paper.

\section{Preliminaries}
In this section, we  introduce some definitions and known results which will be needed in the paper.

\subsection{Ferrers diagram rank metric codes}

Let $\mathbb{F}^{m\times n}_{q}$ be the set of all $m\times n$ matrices over $\mathbb{F}_{q} $. Let  $A$ and $B$ in $\mathbb{F}^{m\times n}_{q}$, the rank metric $d_{R}\left( A,B\right) $ of $A$ and $B$ is defined by ${\rm rank}\left( A-B\right) $. An $\left( m \times n ,M,\delta \right)_{q} $ rank metric code $\mathcal{C}$ is a subset of $\mathbb F^{m\times n}_{q}$ with cardinality $M$ and
$ d_{R}(A,B)\geq \delta $ for all $ A,B \in \mathcal{C} $ with $ A \neq B.$

For an $\left( m \times n ,M,\delta \right)_{q} $ rank-metric code $\mathcal{C}$,  the Singleton-like upper bound (\cite{gabidulin1985theory,delsarte1978bilinear}) implies that

\hskip7cm $\textstyle{M \leq q^{\max\left\lbrace m,n\right\rbrace \left(  \min\left\lbrace m,n\right\rbrace -\delta+1\right)  }}.$
\\
If the equation holds, then $\mathcal{C}$ is called a maximum rank metric  (MRD) code. Furthermore, if $\mathcal{C}$ is a $k$-dimensional $\mathbb{F}_{q}$-linear subspace of $\mathbb{F}^{m\times n}_{q}$, then  it is said to be linear, and is denoted by $\left[ m\times n,k,\delta\right]_{q} $. Linear MRD codes were  constructed for all possible parameters (\cite{gabidulin1985theory,delsarte1978bilinear}). 


\begin{theo}\label{mrtheo1}\emph{\cite[Definition 3]{silva2008rank}}
Let $n \geq 2k$ and $\mathcal{D}$ be a $\left[ k\times \left( n-k\right) , \delta\right]_{q}$-MRD code. Then the lifted MRD code
$\mathscr{C}=\left\lbrace rowspace ( I_{k} \mid  A )  : A\in \mathcal{D} \right\rbrace $
is an $( n , q^{( n-k) ( k-\delta+1) } , 2\delta , k ) _{q}$-CDC, where $I_{k}$ is the $k\times k$ identity matrix.
\end{theo}

In the following, we denote by  $\overline{A}_{q}(n,2\delta,k)$ the maximum size of an $(n,2\delta,k)_{q}$ CDC  which contains a lifted MRD code as a subset.

Ferrers diagrams are used to represent partitions by arrays of dots and empty entries. They are defined as follows.
\begin{Def}\label{mrDef2}\emph{\cite{etzion2009error}}
Given positive integers $m$ and $n$, an $m\times n $ Ferrers diagram $\mathcal{F}$ is an  $m\times n $ array of dots and empty entries satisfying the following conditions$: $

\begin{itemize}[itemsep=0pt,topsep=0pt]
	
	\item the number of dots in each row is at most the number of dots in the previous row;
	
	\item all dots are shifted to the right;
	
	\item the first row has $n$ dots and the rightmost column has $m$ dots.
	
\end{itemize}
\end{Def}

\noindent Let $\mathcal{F}$ be an $m\times n $ Ferrers diagram. For $1\leq i \leq n$, let  $\gamma_{i}$ be the number of dots in the $i$th column of $\mathcal{F}$. Then $\mathcal{F}$ can be written as $\mathcal{F}=\left[ \gamma_{1}, \gamma_{2},\dots \dots \dots , \gamma_{n}\right]$. The transposed Ferrers diagram of $\mathcal{F}$ is defined by $\mathcal{F}^{t}=\left[ \rho_{m}, \rho_{m-1},\dots \dots \dots , \rho_{1}\right]$, where $\rho_{i}$  is the number of dots in the $i$th row of $\mathcal{F}$ for each $1\leq i\leq m$. An $m\times n $ Ferrers diagram  is said to be full if it has $mn$ dots.

\begin{exe}
Let $\mathcal{F}=\left[ 1,  2,  4 \right]$. Then
\begin{center}

	\renewcommand{\arraystretch}{0.58}
	\setlength{\arraycolsep}{0.22em}
	$\begin{matrix}
		\mathcal{F} =\bullet & \bullet &\bullet\\
		& \bullet& \bullet &\\
		& & \bullet&  \\
		& & \bullet&  \\
		
	\end{matrix}$
	\hskip2.5cm
	$\begin{matrix}
		\mathcal{F}^{t} =\bullet & \bullet &\bullet &\bullet \\
		& & \bullet& \bullet &\\
		& & & \bullet&  \\

	\end{matrix}$
\end{center}
\end{exe}

\begin{Def}\label{mrDef3}\emph{\cite{etzion2009error}}
An $[m\times n, k,\delta]_q$ rank metric code $\mathcal{C}$ is  called an $\left[ \mathcal{F},k,\delta \right]_{q} $ Ferrers diagram rank metric code (FDRMC) if for each $U \in \mathcal C$, all the entries  of $U$ not in $\mathcal F$ are zeros.
\end{Def}

The FDRMC of a full Ferrers diagram $\mathcal{F}$ is in fact a rank metric code. Note that if there exists an $\left[ \mathcal{F},k,\delta \right]_{q} $ code, then so does an $\left[ \mathcal{F}^{t},k,\delta \right]_{q} $ code. For given $\mathcal{F}$ and $\delta$, the largest possible dimension of an $\left[ \mathcal{F},k,\delta \right]_{q} $ code is denoted by $\dim\left( \mathcal{F},\delta \right) $. The following lemma gives an upper bound of $\dim\left( \mathcal{F},\delta \right) $.

\begin{lem}\label{mrlem1}\emph{\cite[Theorem 1]{etzion2009error}}
Let $\mathcal F$ be an Ferrers diagram and $\delta $ be a positive integer. For $0\leq i\leq \delta-1 $, denote by $\nu_i$ the number of dots in $\mathcal{F}$ after removing the first $i$ rows and the rightmost $\delta-1-i$ columns. Then
$\dim\left( \mathcal{F},\delta \right) \leq \min\{\nu_i\mid 0\leq i\leq \delta-1\}$.
\end{lem}

Denote  by  $v_{\min}(\mathcal F,\delta)$ the number $\min\ \{\nu_i\mid 0\leq i\leq \delta-1\}$.  Codes that attain  this upper bound are said to be optimal. Construction of optimal FDRMCs has attracted much attentions in recent years, please see  \cite{antrobus2019maximal,etzion2016optimal,etzion2009error,liu2019constructions,liu2023optimal,pratihar2023constructions,zhang2019constructions}  for details. Here we list several known classes of  optimal FDRMCs for later use.

\begin{lem}\label{firstcons}
\emph{\cite[Theorem 3]{etzion2016optimal}}
Let  $\mathcal{F}$ be  an $m\times n  $ Ferrers diagram with $m\!\geq \!n$. Suppose that  each of the rightmost $\delta\!-\!1$ columns of $\mathcal{F}$ has at least $n$ dots. Then there exists an optimal $\left[ \mathcal{F},\sum_{i=1}^{n-\delta+1}\gamma_{i},\delta \right]_{q} $ code for any prime power $q$.
\end{lem}

\begin{lem}\label{stcons}\emph{\cite[Theorem $\uppercase\expandafter{\mathrm{III} }.6.$]{antrobus2019maximal}}
Let  $\mathcal{F}=[\gamma_1,\gamma_2,\cdots,\gamma_n]$ be an $m\times n $ Ferrers diagram with $m\geq n$. Let $2\leq \delta\leq n$, $l=n-\delta+1$ and $\varepsilon$ be the number of dots missing in the rightmost $\delta-1$ columns of $\mathcal{F}$. Suppose that
$\gamma_{i}\leq \gamma_{l+1}-\varepsilon \left( l+1-i \right) $ for $1\leq i\leq l$.
Then there exists an optimal $\left[ \mathcal{F},\sum_{i=1}^{l}\gamma_{i},\delta \right]_{q} $ code for any prime power $q$.

\end{lem}

\begin{lem}\label{gecons}\emph{\cite[Theorem 3.2]{zhang2019constructions}}
Let $1=t_{0}< t_{1}< t_{2}< \dots < t_{l}\leq m$ be integers such that $t_{1} \mid  t_{2} \mid \dots \mid t_{l}$. Let $n, \delta$ be two integers for which $t_{l-1}< n\leq t_{l}$ and $n-t_{1}+1< \delta\leq n$. Let $\mathcal{F}=[\gamma_1,\gamma_2,\cdots,\gamma_n]$ be  an $m\times n $ Ferrers diagram with $m\geq n$. Suppose that
$\gamma_{n-\delta+1}\leq t_{1}$,
$\gamma_{n-\delta+2}\geq t_{1}$,
and $\gamma_{t_{i}+1}\geq t_{i+1}$ for $1\leq i\leq l-1$.
Then there exists an optimal $\left[ \mathcal{F},\delta \right]_{q} $ code for any prime power $q$.
\end{lem}
\begin{lem}\label{ratlem}\emph{\cite[Theorem 38]{pratihar2023constructions}}
Let $(\mu_{1},\dots, \mu_{p} ) $ be a $p$-tuple of positive integers with $p\geq 1$ and set $\mu =\max_{1\leq i \leq p}\ \mu_{i}.$ For any set of $p$ positive integers $\{k_{1},\dots ,k_{p}\} $ such that $k_{1}< k_{2}< \dots < k_{p}$, consider the $m\times n$ Ferrers diagram
$\mathcal{F}=\big\{\overset{\mu_{1} times}{\overbrace{r_{1}, \cdots,  r_{1}}}\
\overset{\mu_{2} times}{\overbrace{r_{2}, \cdots,  r_{2}}}\
,\cdots,\overset{\mu_{p} times}{\overbrace{r_{p}, \cdots,  r_{p}}}\big\}, $
where $r_{i}=k_{i}\mu $ for $1\leq i \leq p$, $m=r_{p}$ and $n=\sum_{i=1}^{p}\mu_{i}.$ Then for any $d$ with $1\leq d \leq n$, there exists an $m\times n$ optimal $\left[ \mathcal{F},d \right]_{q} $ code for any prime power $q> k_{p}$.
\end{lem}

Finally, we need the following result.
\begin{lem}\label{comcons}\emph{\cite[Theorem 9]{etzion2016optimal}}
Let $\mathcal{F}_{1}$ be an $m_{1}\times n_{1}$ Ferrers diagram, $\mathcal{F}_{2}$ be an $m_{2}\times n_{2}$ Ferrers diagram, and $\mathcal{D}$ be an $m_{3}\times n_{3}$ full Ferrers diagram, where $m_{3}\geq m_{1}$ and $n_{3}\geq n_{2}$. Let

\begin{center}
	$\mathcal{F} = \begin{pmatrix}
		\mathcal{F}_{1}&\mathcal{D}\\
		&\mathcal{F}_{2}
	\end{pmatrix}$
\end{center}
be an $m\times n$ Ferrers diagram, where $m=m_{2}+m_{3}$ and $n=n_{1}+n_{3}$.
If there exists an $\left[ \mathcal{F}_{1},k,\delta_{1} \right]_{q} $ code and an $\left[ \mathcal{F}_{2},k,\delta_{2} \right]_{q} $ code, then there exists an $\left[ \mathcal{F},k,\delta_{1}+\delta_{2} \right]_{q} $ code.
\end{lem}

\subsection{Multilevel construction}

Let $U$ be a $k$-dimensional subspace of $\mathbb{F}_{q}^{n}$.  Then $U$ can be represented by a $k \times n$ generator matrix whose $k$ rows form a basis of $U$. Based on Gaussian elimination, there exists a unique generator matrix $EF\left( U\right) $  of $U$ in reduced row echelon form (RREF).

\begin{Def}\label{mrDef1}\emph{\cite{etzion2009error}}
Let $U\in \mathcal{G}_{q}\left( n,k\right) $ and $EF\left( U\right) $ be the RREF of $U$. The identifying vector $\nu \left( U\right) $ of $U$ is a binary row vector of length $n$ and weight $k$, where the $k$ ones in  $\nu \left( U\right) $ are in the positions  where $EF\left( U\right) $ has the leading ones.
\end{Def}

\begin{exe} Let  $U$ be  a $3$-dimensional in $\mathbb{F}_{2}^{5}$ with the generator matrix in RREF
\begin{center}
	$EF\left( U\right)=\begin{pmatrix}
		1&1&0&0&0\\
		0&0&1&0&1\\
		0&0&0&1&0
	\end{pmatrix}$
\end{center}
\noindent The columns of $EF\left( U\right) $  that have the leading ones are columns $1$, $3$ and $4$. Thus the  identifying vector of $U$ is $\nu \left( U\right)=\left( 10110\right)  $.
\end{exe}

The Ferrers tableaux form  $\mathcal{F}\left( U\right) $ of $U$ is obtained by first removing all $zeros$ to the left of the leading coefficient of each row of $EF\left( U\right) $,  then removing the columns that contain the leading coefficients,  and finally moving all the remaining entries to the right.  The Ferrers diagram  $\mathcal{F}_{U}$ of $U$ is obtained by replacing entries of $\mathcal{F}\left( U\right) $ with $\bullet$.

The following two lemmas play a crucial role in the multilevel construction.

\begin{lem}\label{mrlem2.10}\emph{\cite[Lemma 2]{etzion2009error}}
Let $U,V\in \mathcal{P}_{q}\left( n\right)$. Suppose that $\nu(U)$ and $\nu(V)$ is  the identifying vector of $U$ and $V$ respectively.  Then
$d_{S}\left( U,V\right) \geq d_{H}\left(\nu \left( U\right),\nu \left( V\right)\right)  $
\end{lem}

\begin{lem}\label{mrlem2.11}\emph{\cite{etzion2009error}}
For $U,V\in \mathcal{P}_{q}\left( n\right)$, if $\nu \left( U\right)=\nu \left( V\right) $, then $d_{S}\left( U , V\right) = 2d_{R}\left( EF\left( U \right) ,EF\left( V\right) \right) $.
\end{lem}

{\it The procedure of the  multilevel construction} \cite{etzion2009error}:
first  choose an $(n,2\delta,k)_2$ constant weight code $C$; then for each $c\in C$ with  an Ferrers diagram $\mathcal F_c$, construct an  $[\mathcal F_c,\rho_c,\delta]_q$  FDRMC $\mathcal C_c$ and then an $(n,q^{\rho_c},2\delta,k)_q$ lifted FDRMC $\mathscr C_c$;  finally, the union
$\bigcup\limits_{c\in C}\mathscr C_c$
is an $\Big( n, \sum\limits_{c\in C}q^{\rho_c}, 2\delta , k \Big) _{q}$-CDC.

\medskip
Lemma \ref{mrlem2.10} gives  a sufficient condition for finding identifying vectors in the multilevel construction. For some pairs of identifying vectors with Hamming distance $\delta-\delta_{1}$,  we need to use appropriate lifted Ferrers diagram rank-metric codes to ensure that the final subspace distance of the code will be $\delta$. This was done in {\cite{silberstein2015subspace}}  by introducing the notation of  quasi-pending blocks in an Ferrers diagram, which generalized the pending dots defined in \cite{trautmann2010new}.

\begin{Def}\label{mrDef2.14}\emph{\cite[Definition 14]{silberstein2015subspace}}
Let $\mathcal{F}$ be an $m\times n$ Ferrers diagram and $n_{1}< n$ , $m_{1}< m$. If the $\left( m_{1}+1\right) $st row of $\mathcal{F}$ has less dots than the $m_{1}$th row of $\mathcal{F}$ and at most $n-n_{1}$ dots, then the $n_{1}$ leftmost  columns of  $\mathcal{F}$ are called a quasi-pending block of size $m_{1}\times n_{1}$.
\end{Def}	

\begin{lem}\label{lempendingblock}\emph{\cite[Theorem 15]{silberstein2015subspace}}
For $U,V\in \mathcal{G}_{q}\left( k, n\right)$, such that  $d_{H}\left( \nu \left( U\right),\nu \left( V\right) \right) =2\delta$ and $\mathcal{F}_{U},\mathcal{F}_{V}$ have a quasi-pending block of size $m_{1}\times n_{1}$ in the same position. Denote by $B_{U}$ and $B_{V}$ the submatrix of $\mathcal{F}_{U}$ and $\mathcal{F}_{V}$ corresponding to the quasi-pending block, respectively. Then $d_{S}\left( U , V\right) \geq 2\delta +2rank\left( B_{U}-B_{V}\right) $.

\end{lem}

At the end of this section, we present some definitions and known  results in graph theory.
Let $K_{m}$ be the complete graph with $m$ vertices. A matching of $K_{m}$ is a set of pairwise non-adjacent edges of $K_{m}$. A perfect (resp.nearly perfect) matching is a matching that covers all (resp.all but one) vertices of $K_{m}$. A one-factorization (resp. near one-factorization) of $K_{m}$ is a partition of all edges into perfect (resp. nearly perfect) matching of $K_{m}$. If one labels all vertices of $K_{m}$ with the number from $1 , 2 , \dots , m$, then one can easily see a one-to-one correspondence between the edges of graph and the vectors of weight 2 in $\mathbb{F}_{2}^{m}$ by assigning two ones in the coordinates labeled by the numbers of the two vertices in the graph which are connected by the corresponding edge.

The following lemma is a well-known result in graph theory.

\begin{lem}\label{lempar}\emph{\cite{van2001course}}
Let $D$ be the set of all binary vectors of length $m$ and weight $2$. 
\begin{itemize}[itemsep=0pt,topsep=0pt]
	\item[$(1)$] If $m$ is even, then $D$ can be partitioned into $m-1$ classes, each class has $\frac{m}{2}$ vectors with pairwise disjoint positions of $ones$;

	\item[$(2)$] If $m$ is odd, then $D$ can be partitioned into $m$ classes, each class has $\frac{m-1}{2}$ vectors with pairwise disjoint positions of $ones$.
	
\end{itemize}

\end{lem}

\section{The constructions}
In this section, we focus on constructing  $((n+3)t,4t,3t)_q$  CDCs by the
multilevel construction for $n\geq 5$ and  $t$ is even. As applications, we provide  new lower bounds of $A_q(n,8,6)$  for $16\leq n\leq 19$ at the end of this section.

\subsection{Constructions of $\big((n+3)t,4t,3t\big)_q$ constant dimension codes}

Let $(n,t)\in\mathbb N^2$ with $n\geq 5$ and $t\geq 2$. Suppose that $t$ is even.  Let  $\nu=(\nu_1,\nu_2,\cdots,\nu_n)\in \mathbb F_2^n$. We define $w\in\mathbb F_2^{nt}$ as follows:
$$w_{i+(k-1)t}=\left\{\begin{array}{cc}
1,&\mbox{if}\ \nu_k=1;\\
0,&\mbox{if}\  \nu_k=0,
\end{array}\right.$$
where $1\leq i\leq t, 1\leq k\leq n.$ The new vector  $w$ is  called {\bf  the extension  of $\nu$ with  ratio $t$ } and is denoted by $E_t(\nu)$. For example, let
$\nu=(10 010)\in\mathbb F_2^5$, then the  extension of $\nu$ with  ratio $3$ is $w=(111\ 000\ 000\ 111\ 000).$

Set $s=n-1$ if $n$ is even and $s=n$ if $n$
is odd. By  Lemma \ref{lempar}, we can divide  the set of all binary vectors of length $n$ and weight $2$ into $s$ classes
$P_1,P_2,\cdots,P_s,$ where each class is an $(n,4,2)_2$ constant weight code.
Let $1\leq i\leq s$.
Define $Q_i=\{w\in \mathbb F_2^{nt}\mid w=E_t(\nu), \nu\in P_i\}.$ It is easy to see that $Q_i$ is an $(nt,4t,2t)_2$ constant weight code for $1\leq i\leq s$  such that $\bigcup\limits_{i=1}^sQ_i$ is an $(nt,2t,2t)_2$ constant weight code in $\mathbb F_2^{nt}.$ Now we can give our constructions. For convenience, set $a=\frac{t}{2}$.

{\bf  Identifying vectors:}
\begin{itemize}[itemsep=0pt,topsep=0pt]
\item {\it The first class:} Define $\mathcal A=\big\{(\underset{t}{\underbrace{1\cdots 1}}\ \underset{t}{\underbrace{1\cdots 1}}\ \underset{t}{\underbrace{1\cdots 1}}\    \|     \ \underset{nt}{\underbrace{0\cdots\cdots 0}}\ )\big\}.$
\item{\it The second class:}  For $1\leq i\leq s$, define $\mathcal B_i=\big\{\big(\underset{t}{\underbrace{1\cdots 1}}\ \underset{t}{\underbrace{0\cdots 0}}
\ \underset{t}{\underbrace{0\cdots 0}}\    \| \    w\ \big ) \mid w\in Q_i\big\}  $,\
set $\mathcal B=\bigcup\limits_{i=1}^s\mathcal B_i.$

\item{\it The third class:}
Let $\alpha=\big(\underset{a}{\underbrace{1\cdots 1}} \underset{a}{\underbrace{0\cdots 0}}\big)$, $  \beta=\big(\underset{a}{\underbrace{0\cdots 0}} \underset{a}{\underbrace{1\cdots 1}}\big)$, $\boldsymbol{0}=\big(\underset{(n-4)t}{\underbrace{0\cdots 0}} \big)$. Obviously, $d_{H}\left( \alpha , \beta \right) =t.$ We define the following  vectors:

$u_1=\big(\alpha \   \|  \   \alpha \   \|  \   \alpha \   \|  \   \alpha \   \|  \   \boldsymbol{0}\big)
$,
\ $u_2=\big(\beta \   \|  \   \beta \   \|  \   \alpha \   \|  \   \alpha \   \|  \   \boldsymbol{0}\big)
$,
\ $u_3=\big(\beta \   \|  \   \alpha \   \|  \   \beta \   \|  \   \alpha \   \|  \  \boldsymbol{0}\big)
$,
\ $u_4=\big(\beta \   \|  \   \alpha \   \|  \   \alpha \   \|  \   \beta \   \|  \   \boldsymbol{0}\big)
$,

$u_5=\big(\alpha \   \|  \   \beta \   \|  \   \alpha \   \|  \   \beta \   \|  \   \boldsymbol{0}\big)
$,
\ $u_6=\big(\alpha \   \|  \   \alpha \   \|  \   \beta \   \|  \   \beta \   \|  \   \boldsymbol{0}\big)
$,
\ $u_7=\big(\alpha \   \|  \   \beta \   \|  \   \beta \   \|  \   \alpha \   \|  \   \boldsymbol{0}\big)
$,
\ $u_8=\big(\beta \   \|  \   \beta \   \|  \   \beta \   \|  \   \beta \   \|  \   \boldsymbol{0}\big)
$.

\medskip
Set
$\mathcal C=\left\{\begin{array}{ll}
	\Big\{\big(\underset{t}{\underbrace{0\cdots 0}}\ \underset{t}{\underbrace{1\cdots 1}}\
	\underset{t}{\underbrace{0\cdots 0}}\ \   \|  \  \  u\ \big ) \mid u\in \{u_1,u_2,\cdots, u_8\}\Big\},&\mbox{ if} \  q^t\geq 8;\\
	\Big\{\big(\underset{t}{\underbrace{0\cdots 0}}\ \underset{t}{\underbrace{1\cdots 1}}\
	\underset{t}{\underbrace{0\cdots 0}}\ \   \|  \  \  u\ \big ) \mid u\in \{u_1,u_2,u_3,u_4\}\Big\},& \ \mbox{if } q^t=4.
\end{array}\right.$

\end{itemize}

\bigskip
{\bf   Ferrers tableaux forms and quasi-pending blocks:}

\begin{itemize}[itemsep=0pt,topsep=0pt]

\item	All the Ferrers diagrams which correspond to the identifying vectors from $\mathcal B$ have a common quasi-pending block in
the  leftmost $2t$ columns, which is a full diagram of size $t\times 2t.$  Let $\mathcal M$ be a $[t\times 2t,q^{2t},t]_{q}$ MRD code. We fill the same matrix  of $\mathcal M$ in the quasi-pending block for each vector in the same class. For the identifying vectors from different classes, we fill  different matrices of   $\mathcal M$. For this purpose,   we need the condition that $q^{2t}\geq s.$ For the remaining dots of Ferrers diagrams, we construct  FDRMCs with distance $2t$ and then lift them.

\item	All the Ferrers diagrams which correspond to the identifying vectors from $\mathcal C$ have a common quasi-pending block in
the leftmost $t$ columns, which is a full diagram of size $t\times t.$ Let $\mathcal N$ be a $[t\times t,q^{t},t]_{q}$ MRD code. We fill the same matrix of $\mathcal N$ in this block for the same  vector. For different  identifying vectors, we fill  different matrices of  $\mathcal N$. For this purpose,  we need the condition that $q^{t}\geq |\mathcal C|.$ For the remaining dots of Ferrers diagrams, we construct  FDRMCs with distance $2t$ and then  lift them.
\end{itemize}

{\bf Codes: }\  Let $\mathscr A$ be the lifted  $[3t\times nt,2t]_q$ MRD code constructed from the single vector in $\mathcal A$. Let $\mathscr B,\mathscr C$ be the union of the lifted codes corresponding to  $\mathcal B,\mathcal C$ respectively.  Set $\mathscr D=\mathscr A\cup \mathscr B\cup\mathscr C$. Then we get the following result.

\begin{theo} \label{theo1}The subspace  code $\mathscr D$   is an $\big((n+3)t,4t,3t\big)_{q}$ constant dimension code.
\end{theo}
\begin{proof}
The proof is motivated by  \cite{silberstein2015subspace}. It suffices to check the distance.
Let $X$ and $Y$ be in $\mathscr D$.
\begin{itemize}[itemsep=0pt,topsep=0pt]
	\item 	If $X\in \mathscr A$ and $Y\in\mathscr A$, then $d_S(X,Y)\geq 4t$ since $\mathscr A$ is an $((n+3)t,4t,3t)_{q}$ CDC.
	\item If $X\in \mathscr A$ and $Y\in \mathscr B\cup \mathscr C$, then $d_S(X,Y)\geq 4t$ since $d_H(v(X),v(Y))\geq 4t$.
	\item  If $X\in \mathscr B$ and $Y\in\mathscr C$, then $d_S(X,Y)\geq 4t$ since $d_H(v(X),v(Y))\geq 4t$.
	\item Suppose that $(X,Y)\in \mathscr B^2$ such that $v(X)\in\mathcal B_i$ and $v(Y)\in \mathcal B_j$.
	
	If $v(X)=v(Y)$, then $d_S(X,Y)\geq 4t$ by construction.
	Now suppose that $v(X)\neq v(Y)$.
	
	If $i=j$, then $d_H(v(X),v(Y))\geq 4t$. It follows that  $d_S(X,Y)\geq 4t$ .
	
	If $i\neq j$, then $d_H(v(X),v(Y))\geq 2t$. Let $B_X$ and $B_Y$ be the matrix filled in the pending block corresponding to $v(X)$ and $V(Y)$, respectively. Then ${\rm rank}(B_X-B_Y)\geq t$ by construction. It follows from Lemma \ref{lempendingblock} that $d_S(X,Y)\geq 2{\rm rank}(B_X-B_Y)+d_H(v(X),v(Y))\geq 4t.$
	
	\item Suppose that $(X,Y)\in \mathscr C^2$.	If $v(X)=v(Y)$, then $d_S(X,Y)\geq 4t$ by construction.
	Now suppose that $v(X)\neq v(Y)$.
	Then $d_H(v(X),v(Y))\geq 2t$. Let $C_X$ and $C_Y$ be  the matrix filled in the quasi-pending block corresponding to $v(X)$ and $v(Y)$, respectively. Then ${\rm rank}(C_X-C_Y)\geq t$ by construction. It follows from Lemma \ref{lempendingblock} that $d_S(X,Y)\geq 2{\rm rank}(C_X-C_Y)+d_H(v(X),v(Y))\geq 4t.$
\end{itemize}

It follows that $d_S(X,Y)\geq 4t$ for any two different subspaces $X$ and $Y$ in $\mathscr D.$ Consequently, $\mathscr D$  is an $\big((n+3)t,4t,3t\big)_{q}$ constant dimension code.
\end{proof}

\subsection{Calculation of the cardinality of $\mathscr D$}
In this subsection, we calculate the cardinality of  the constant dimension code $\mathscr D$.
Obviously, $|\mathscr A|=q^{nt(t+1)}$ and $|\mathscr D|=|\mathscr A|+|\mathscr B|+|\mathscr C|$.

\subsubsection{Calculation of the cardinality  of $\mathscr B$}
Let $1\leq i<j\leq n$. For convenience,  we  denote by   $u_{i,j}$
the extension of the vector $(10 0\|u)$ with ratio $t$, where $u\in\mathbb F_2^n$ with sup$(u)= \{i,j\}$.
Denote by $\mathcal F_{i,j}$ the subdiagram of the  Ferrers diagram of $u_{i,j}$ by deleting the leftmost $2t$ columns and by $\mathcal C_{i,j}$ the FDRMC constructed from $\mathcal F_{i,j}$. It is easy to check that  $\mathcal F_{i,j}=[\underset{(i-1)t}{\underbrace{t,\cdots,t}},\underset{(j-i-1)t}{\underbrace{2t,\cdots,2t}},
\underset{(n-j)t}{\underbrace{3t\cdots,3t}}],$  which also can  be  represented   as follows:

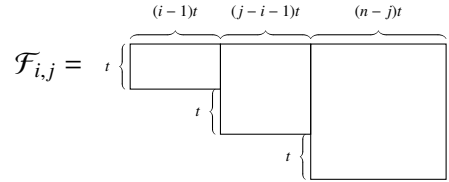
\begin{center}
\vskip-0.5cm
\begin{tikzpicture}[scale=0.6]
	
	
	\draw[]  (0,2) rectangle (2,3);
	\draw[] (2,1) rectangle (4,3);
	\draw[] (4,0) rectangle (7,3);

	\draw [decorate,
	decoration = {calligraphic brace, 
		raise=2pt, 
		aspect=0.5, 
		amplitude=3pt 
	}] (0,3) --  (2,3);

	\draw [decorate,
	decoration = {calligraphic brace, 
		raise=2pt, 
		aspect=0.5, 
		amplitude=3pt 
	}] (2,3) --  (4,3);

	\draw [decorate,
	decoration = {calligraphic brace, 
		raise=2pt, 
		aspect=0.5, 
		amplitude=3pt 
	}] (4,3) --  (7,3);

	\draw [decorate,
	decoration = {calligraphic brace, 
		raise=2pt, 
		aspect=0.5, 
		amplitude=2pt 
	}] (0,2) --  (0,3);

	\draw [decorate,
	decoration = {calligraphic brace, 
		raise=2pt, 
		aspect=0.5, 
		amplitude=2pt 
	}] (2,1) --  (2,2);

	\draw [decorate,
	decoration = {calligraphic brace, 
		raise=2pt, 
		aspect=0.5, 
		amplitude=2pt 
	}] (4,0) --  (4,1);

	\node at  (1,3.7) {\tiny $(i-1)t$};
	
	\node at  (3,3.7) {\tiny $(j-i-1)t$};
	
	\node at  (5.5,3.7) {\tiny $(n-j)t$};

	\node at  (-0.5,2.5){\tiny $t$};
	
	\node at  (1.5,1.5) {\tiny $t$};
	
	\node at  (3.5,0.5) {\tiny $t$};

	\node at  (-1.8,2.5) {\small$\mathcal{F}_{i,j}=$};
\end{tikzpicture}

\end{center}

\noindent$\big(1\big)$ Case $1\leq i\leq n-4.$
\begin{itemize}
\item[(i)]	Let $2\leq j\leq n$ and   $\mathcal F'_{1,j}$ be the rightmost $(n-2)t$ columns of $\mathcal F_{1,j}$. It
is easy to check that  $v_{\min}(\mathcal F'_{1,j},2t)=v_{\min}(\mathcal F_{1,j},2t)=(n- 2)t+(n-j)t^2$.  Since
the topmost $2t-1$ rows of $\mathcal F'_{1,j}$ have $(n-2)t\geq 3t$ dots,  we get from Lemma \ref{firstcons} that there exists an optimal $[\mathcal F'_{1,\!j},{(n\!-\!2)t\!+\!(n\!-\!j)t^2},2t]_{q}$ FDRMC. Thus,  there exists an optimal $[\mathcal F_{1,\!j},{(n\!-\!2)t\!+\!(n\!-\!j)t^2},2t]_{q}$ FDRMC.
\item[(ii)]  Let  $2\leq i\leq n-4$ and $i<j\leq n.$ One can check that there exists an optimal $[\mathcal F_{i,j},{(n-i-1)t+(n-j)t^2},2t]_{q}$ FDRMC similarly.
\end{itemize}
\noindent$\big(2\big)$ Case $i=n-1.$
Since $\mathcal F_{n-1,n}=[\underset{(n-2)t}{\underbrace{t,\cdots,t}}]$, we know that $v_{\min}(\mathcal F_{n-1,n},2t)=0.$

\noindent$\big(3\big)$ Case $i=n-2.$
\begin{itemize}
\item[(i)]
Since $\mathcal F_{n-2,n-1}=[\underset{(n-3)t}{\underbrace{t,\cdots,t}},\underset{t}{\underbrace{3t,\cdots,3t}}]$, we know that $v_{\min}(\mathcal F_{n-2,n-1},2t)=2t.$ Let $ \mathcal F_{n-2,n-1}'$  be the rightmost $3t$ columns of  $\mathcal F_{n-2,n-1}$. Then $ \mathcal F_{n-2,n-1}'$ can be written as
$ \mathcal F_{n-2,n-1}'=\left(\begin{smallmatrix}
	\mathcal{F}'_{1}&\mathcal{F}'_{2}\\
	&\mathcal{F}'_{3}
\end{smallmatrix}\right)$, where both  $(\mathcal F'_{1})^{t}$ and $\mathcal F'_{3}$ are $2t\times t$ full Ferrers diagram and $\mathcal F'_{2}$ is a $t\times t$ full Ferrers diagram. Moreover, $v_{\min}(\mathcal F'_{n-2,n-1},2t)=v_{\min}(\mathcal F_{n-2,n-1},2t)=2t.$ By Lemma \ref{firstcons}, there exists an optimal $[(\mathcal F_1')^{t},2t,t]_{q}$ FDRMC and an  optimal $[\mathcal F'_{3},2t,t]_{q} $ FDRMC. It follows from Lemma \ref{comcons} that there exists an optimal $[\mathcal F'_{n-2,n-1},2t,2t]_{q}$ FDRMC. Thus,
there exists an optimal $[\mathcal F_{n-2,n-1},2t,2t]_{q}$ FDRMC.

\item[(ii)] Since $\mathcal F_{n-2,n}=[\underset{(n-3)t}{\underbrace{t,\cdots,t}},\underset{t}{\underbrace{2t,\cdots,2t}}]$, we know that $v_{\min}(\mathcal F_{n-2,n},2t)=t.$ Let $\mathcal F'_{n-2,n}$ be the rightmost $2t$ columns of  $\mathcal F_{n-2,n}$.
Then $v_{\min}(\mathcal F'_{n-2,n},2t)=v_{\min}(\mathcal F_{n-2,n},2t)=t.$
Similar to the proof of (i), we get  from Lemma \ref{comcons} that there exists an optimal $[\mathcal F'_{n-2,n},t,2t]_{q}$ FDRMC. Then
there exists an optimal $[\mathcal F_{n-2,n},t,2t]_{q}$ FDRMC.
\end{itemize}

\noindent$\big(4\big)$ Case $i=n-3$.
\begin{itemize}
\item[(i)]
Note that $\mathcal F_{n-3,n-2}=[\underset{(n-4)t}{\underbrace{t,\cdots,t}},\underset{2t}{\underbrace{3t,\cdots,3t}}]$.

\begin{itemize}
	\item[$\bullet $] If $n\geq 6,$ then we set   $\mathcal F'_{n-3,n-2}$
	be the rightmost $4t$ columns  of $\mathcal F_{n-3,n-2}$. It is easy to check that
	$v_{\min}(\mathcal F'_{n-3,n-2},2t)=v_{\min}(\mathcal F_{n-3,n-2},2t)=2t^2+2t.$
	It follows from Lemma \ref{gecons} that there exists an optimal $[\mathcal F_{n-3,n-2},2t^2+2t,2t]_{q}$ FDRMC.
	\item[$\bullet $] If $n=5$, then $\mathcal F_{n-3,n-2}=[\underset{t}{\underbrace{t,\cdots,t}},\underset{2t}{\underbrace{3t,\cdots,3t}}]$.
	It follows from Lemma \ref{firstcons} that there exists an optimal $[\mathcal F_{n-3,n-2},t^2+3t,2t]_{q}$ FDRMC.
\end{itemize}

\item[(ii)] Now we consider the case  $\mathcal F_{n-3,n-1}=[\underset{(n-4)t}{\underbrace{t,\cdots,t}},\underset{t}{\underbrace{2t,\cdots,2t}},\underset{t}{\underbrace{3t,\cdots,3t}}]$.

If $n\geq 6,$  then we set  $\mathcal F''_{n-3,n-1}=[\underset{2
	t}{\underbrace{t,\cdots,t}},\underset{t}{\underbrace{2t,\cdots,2t}},\underset{t}{\underbrace{3t,\cdots,3t}}]$
be the rightmost $4t$ columns  of $\mathcal F_{n-3,n-1}$. It is easy to check that
$v_{\min}(\mathcal F''_{n-3,n-1},2t)=v_{\min}(\mathcal F_{n-3,n-1},2t)=t^2+2t.$ Note that  the transpose of  $\mathcal F''_{n-3,n-1}$ is $[\underset{
	t}{\underbrace{t,\cdots,t}},\underset{t}{\underbrace{2t,\cdots,2t}},\underset{t}{\underbrace{4t,\cdots,4t}}]$.
It follows from Lemma \ref{gecons} that  there exists an optimal $[\mathcal F''_{n-3,n-1},t^2+2t,2t]_{q}$ FDRMC and thus exists an optimal $[\mathcal F_{n-3,n-1},t^2+2t,2t]_{q}$ FDRMC. 

If $n=5,$  then $\mathcal F_{n-3,n-1}=[\underset{
	t}{\underbrace{t,\cdots,t}},\underset{t}{\underbrace{2t,\cdots,2t}},\underset{t}{\underbrace{3t,\cdots,3t}}]$ and
$v_{\min}(\mathcal F_{n-3,n-1},2t)=t^2+2t.$

\begin{itemize}
	
	\item[$\bullet$]	If $t=2$, then $\mathcal F_{n-3,n-1}=[2,2,4,4,6,6]$. It follows from \cite[Example III.16]{antrobus2019maximal} that there exists
	an optimal $[\mathcal F_{n-3,n-1},8,4]_q$ FDRMC.
	
	\item [$\bullet$]	If $q\geq 4$, then it follows from Lemma \ref{ratlem} that there exists an  optimal $[\mathcal F_{n-3,n-1},t^2+2t,2t]_{q}$
	FDRMC.
	
	\item[$\bullet$] In general, we don't know whether there exits an optimal $[\mathcal F_{n-3,n-1},2t]_{q}$
	FDRMC. However, we  can construct  an
	$[\mathcal F_{n-3,n-1},2t]_{q}$
	FDRMC with smaller size.
	Let  $\mathcal{C}_{1}$ be a $[t\times t, t]_{q}$ MRD code and $\mathcal{C}_{2}$  a $[2t\times 2t, 2t]_{q}$ MRD code.  Define $\mathcal{C}=\left\lbrace \left[ \begin{smallmatrix}
		A&B\\
		&A
	\end{smallmatrix} \right] \mid  A\in \mathcal{C}_{1}, B \in \mathcal{C}_{2} \right\rbrace $.  Then $\mathcal{C} $ is an $[\mathcal F_{n-3,n-1},3t,2t]_{q}$ FDRMC.
\end{itemize}
\item[(iii)] Finally,  we consider the case  $\mathcal F_{n-3,n}=[\underset{(n-4)t}{\underbrace{t,\cdots,t}},\underset{2t}{\underbrace{2t,\cdots,2t}}]$.
Let $\mathcal F'_{n-3,n}$ be the rightmost $2t$ columns of $\mathcal F_{n-3,n}$. Then $v_{\min}(\mathcal F'_{n-3,n},2t)=v_{\min}
(\mathcal F_{n-3,n},2t)=2t$. It follows from Lemma \ref{firstcons} that there exists an optimal
$[\mathcal F_{n-3,n},2t,2t]_{q}$ FDRMC.
\end{itemize}
Consequently, we get the following result:
\begin{center}

$|\mathscr B|=\left\{\begin{array}{llll}
	\sum\limits_{i=1}^{n\!-\!4}\sum\limits_{j=i\!+\!1}^{n}q^{(n\!-\!i\!-\!1)t\!+\!(n\!-\!j)t^{2}}\!+\!q^{2t^{2}\!+\!2t}\!+\!q^{t^{2}\!+\!2t}\!+\!2q^{2t}\!+\!q^{t}\!+\!1,&\!\!\mbox{ if} \   n\geq6;\\
	q^{18}\!+\!q^{14}\!+\!2q^{10}\!+\!q^{8}\!+\!q^6\!+\!2q^{4}\!+\!q^{2}\!+\!1,&\!\! \ \mbox{if} \ n=5, t=2;\\
	\sum\limits_{j=2}^{5}q^{3t\!+\!(5\!-\!j)t^{2}}\!+\!q^{t^{2}\!+\!3t}\!+\!q^{t^{2}\!+\!2t}\!+\!2q^{2t}\!+\!q^{t}\!+\!1,&\!\! \ \mbox{if} \ n=5, t\neq2, q\geq 4;\\
	\sum\limits_{j=2}^{5}q^{3t\!+\!(5\!-\!j)t^{2}}\!+\!q^{t^{2}\!+\!3t}\!+q^{3t}\!+\!2q^{2t}\!+\!q^{t}\!+\!1,&\!\! \ \mbox{if} \ n=5, t\neq2, q< 4.\\
\end{array}\right.$
\end{center}

\subsubsection{Calculation of the cardinality  of $\mathscr C$}

Let $1\leq i \leq 8$ and $v_{i}=\big(\underset{t}{\underbrace{0\cdots 0}}\ \underset{t}{\underbrace{1\cdots 1}}\
\underset{t}{\underbrace{0\cdots 0}}\ \   \|  \  \  u_{i}\ \big )\in \mathcal{C}$. Denote by $\mathcal F_{i}$ the subdiagram of the  Ferrers diagram of $v_{i}$ by deleting the leftmost $t$ columns and by $\mathscr{C}_{i}$ the lifted  FDRMC corresponding to  $v_{i}$.

\medskip
\noindent$\big(1\big)$ Case $n\geq 6$.

\noindent
\medskip
Let $\mathcal F'_{1}=[\underset{a}{\underbrace{4a,\cdots,4a}},\underset{a}{\underbrace{5a,\cdots,5a}},\underset{2a(n-4)+a}{\underbrace{6a,\cdots,6a}}]$ be the subdiagram of $\mathcal F_{1}$ by deleting the leftmost $a$ columns. It is easy to check that  $v_{\min}(\mathcal F_{1},2t)=v_{\min}(\mathcal F'_{1},2t)$. Since the topmost $2t-\!1$ columns of $\mathcal{F}'_{1}$ have full dots, it follows from Lemma \ref{firstcons} that there exists an optimal $[\mathcal{F}'_{1},(t^{2}+t)n-\frac{13t^{2}}{4}-\frac{5t}{2},2t]_q$ FDRMC. Thus, there exits an optimal $[\mathcal{F}_{1},(t^{2}+t)n-\frac{13t^{2}}{4}-\frac{5t}{2},2t]_q$ FDRMC.  Similarly, one can check that there exists an optimal   $[\mathcal{F}_{i},2t]_q$ FDRMC for each $2\leq i\leq 8.$ By a direct calculation, we get the following result:
\begin{center}
Table 1

\medskip
\begin{tabular}{|cc|cc|cc|cc|}
	\hline
	${\rm vector}$ &  ${\rm size } $ &${\rm vector}$ &  ${\rm size}$ &${\rm vector}$ &  ${\rm size}$ & ${\rm vector}$ & ${\rm size}$  \\
	\hline
	
	$v_{1}$  & $q^{(t^{2}+t)n-\frac{13t^{2}}{4}-\frac{5t}{2}} $&$v_{2}$  & $q^{(t^{2}+t)n-\frac{13t^{2}}{4}-3t} $&$v_{3} $ &$q^{(t^{2}+t)n-\frac{7t^{2}}{2}-\frac{5t}{2}} $&$v_{4} $&  $q^{(t^{2}+t)n-\frac{7t^{2}}{2}-\frac{5t}{2}} $ \\
	\hline		
	$v_{5} $ &  $q^{(t^{2}+t)n-\frac{7t^{2}}{2}-3t}$&$v_{6}$&$q^{(t^{2}+t)n-\frac{15t^{2}}{4}-\frac{5t}{2}}$ &$v_{7}$ &
	$q^{(t^{2}+t)n-\frac{7t^{2}}{2}-3t}$&$v_{8}$&$q^{(t^{2}+t)n-\frac{15t^{2}}{4}-3t}$\\
	\hline
\end{tabular}
\end{center}

\noindent$\big(2\big)$ Case $n=5$ and $t=2$.
\begin{itemize}
\item[$\bullet$] Since $\mathcal{F}_{1}=[3,4,5,6,6,6]$,  we get from  Lemma \ref{firstcons} that there exists an optimal $[\mathcal F_{1},4]_q$ FDRMC, where $v_{\min}(\mathcal F_{1},4)=12$. Thus, $|\mathscr{C}_{1}|=q^{12}$. Similarly, we can construct an   optimal $[\mathcal F_{i},4]_q$ FDRMC and thus  a CDC  $\mathscr{C}_{i}$ with size $q^{10}$ for  each $ i\in\{2,3,5,6,7\}$.

\item[$\bullet$]Since $\mathcal{F}_{8}=[2,3,4,5,6,6]$,  we get from  Lemma \ref{stcons} that there exists an optimal $[\mathcal F_{8},4]_q$ FDRMC, where $v_{\min}(\mathcal F_{8},4)=9$. Thus, $|\mathscr{C}_{8}|=q^{9}$.

\item[$\bullet$]  Note that $\mathcal{F}_{4}=[2,4,5,5,6,6]$. Let $\mathcal{F}'_{4}=[4,5,5,5,5]$ be the subdiagram of $\mathcal{F}_{4}$ by deleting the leftmost column and last row. It follows from Lemma \ref{firstcons} that there exists an optimal $[\mathcal{F}'_{4},9,4]_q$ FDRMC. Then  there exists an $[\mathcal{F}_{4},9,4]_q$ FDRMC. Hence, $|\mathscr{C}_{4}|=q^{9}$.
\end{itemize}

\noindent$\big(3\big)$ Case $n=5$ and $t\neq2$.
\begin{itemize}
\item[$\bullet$]In this case, $\mathcal F_{1}=[\underset{a}{\underbrace{3a,\cdots,3a}},\underset{a}{\underbrace{4a,\cdots,4a}},\underset{a}{\underbrace{5a,\cdots,5a}},\underset{3a}{\underbrace{6a,\cdots,6a}}].$ Let $\mathcal{F}'_{1}=[\underset{a}{\underbrace{4a,\cdots,4a}},\underset{4a}{\underbrace{5a,\cdots,5a}}]$ be the subdiagram of $\mathcal{F}_{1}$ by deleting the leftmost $a$ columns and last $a$ rows. It follows from Lemma \ref{firstcons} that there exists an optimal $[\mathcal{F}'_{1},t^{2}+\frac{5t}{2},2t]_q$ FDRMC. Then  there exists an $[\mathcal{F}_{1},t^{2}+\frac{5t}{2},2t]_q$ FDRMC. Thus $|\mathscr{C}_{1}|=q^{t^{2}+\frac{5t}{2}}$. We  can construct $\mathscr{C}_{i}$ for  $2\leq i \leq 6$  similarly and get the following result:
\begin{center}
	Table 2
	
	\medskip
	\begin{tabular}{|cc|cc|cc|cc|cc|cc|}
		\hline
		${\rm vector}$ &  ${\rm size}$ &${\rm vector}$ &  ${\rm size}$ &${\rm vector}$ &  ${\rm size}$ &${\rm vector}$ &  ${\rm size}$ &${\rm vector}$ &  ${\rm size}$ & ${\rm vector}$ & ${size}$  \\
		\hline
		
		$v_{1}$  & $q^{t^{2}+\frac{5t}{2}} $&$v_{2}$  & $q^{\frac{3t^{2}}{4}+\frac{5t}{2}} $&$v_{3} $ &$q^{\frac{3t^{2}}{4}+\frac{5t}{2}} $&$v_{4} $&  $q^{t^{2}+\frac{5t}{2}} $  &$v_{5} $ &  $q^{\frac{3t^{2}}{4}+\frac{5t}{2}}$&$v_{6}$&$q^{\frac{3t^{2}}{4}+\frac{5t}{2}}$ \\
		\hline
	\end{tabular}
\end{center}
\item[$\bullet$] $\mathcal F_7$ can be written as $\mathcal{F}_{7}= \left(\begin{smallmatrix}
	\mathcal{F}'_{1}&\mathcal{F}'_{2}\\
	&\mathcal{F}'_{3}
\end{smallmatrix}\right),$ where  both $\mathcal F'_{1}$ and $(\mathcal F'_{3})^{t}$ are $3a\times 2a$ full Ferrers diagrams and $\mathcal F'_{2}$ is a $4a\times 4a$ full Ferrers diagram. Let $\mathcal{C}_{1}$ be a $[3a\times 2a, t]_{q}$ MRD code and $\mathcal{C}_{2}$ a $[4a\times 4a, 2t]_{q}$ MRD code. Define  $\mathcal{C}_7=\left\lbrace  \left[\begin{smallmatrix}
	A&B\\
	&A^{t}
\end{smallmatrix} \right] \mid  A\in \mathcal{C}_{1}, B \in \mathcal{C}_{2} \right\rbrace $. Then $\mathcal{C}_7$   is an $[\mathcal{F}_{7},q^{\frac{7t}{2}},2t]_{q}$ FDRMC. Thus  $|\mathscr{C}_{7}|=q^{\frac{7t}{2}}$.
Similarly, we can construct   an $[\mathcal{F}_{8},3t,2t]_q$ FDRMC. Thus,  $|\mathscr{C}_{8}|=q^{3t}$.

\end{itemize}

Consequently, we get the following:
\begin{center}

$|\mathscr C|=\left\{\begin{array}{lllll}
	q^{6n\!-\!18}\!+\!3q^{6n\!-\!19}\!+\!3q^{6n\!-\!20}\!+\!q^{6n\!-\!21},&\!\!\!\mbox{if}  \ n\geq 6, t=2, q>2;\\
	q^{6n\!-\!18}\!+\!3q^{6n\!-\!19},&\!\!\!\mbox{if}  \ n\geq 6, t=2, q=2;\\
	N,&\!\!\!\!\!\mbox{ if} \   n\geq6, t\neq2;\\
	
	q^{12}\!+\!5q^{10}\!+\!2q^{9},&\!\!\!\mbox{if} \ n=5, t=2, q>2;\\
	q^{12}\!+\!3q^{10},&\!\!\!\mbox{if} \ n=5, t=2, q=2;\\
	2q^{t^{2}\!+\!\frac{5t}{2}}\!+\!4q^{\frac{3t^{2}}{4}\!+\!\frac{5t}{2}}\!+q^{\frac{7t}{2}}\!+\!q^{3t},&\!\!\! \mbox{if} \ n=5, t\neq2;\\
\end{array}\right.$
\end{center}
where  $N=q^{(t^{2}+t)n-\frac{13t^{2}}{4}-\frac{5t}{2}}+q^{(t^{2}+t)n-\frac{13t^{2}}{4}-3t}+2q^{(t^{2}+t)n-\frac{7t^{2}}{2}-\frac{5t}{2}}+2q^{(t^{2}+t)n-\frac{7t^{2}}{2}-3t}+q^{(t^{2}+t)n-\frac{15t^{2}}{4}-\frac{5t}{2}}+q^{(t^{2}+t)n-\frac{15t^{2}}{4}-3t}$.

Now  we can state our main result.

\begin{theo}\label{myth2}
Let $(n, t)\in
\mathbb{N}_{+}^{2}$  and $q$ be a prime power. Suppose that $t$ is even and $q^{2t} \geq s$, where $s=n-1$ for even $n$ (or $s=n$ for odd $n$).
\begin{itemize}
	\item[$(1)$]If $n\geq 6$, then $\overline{A}_{q}((n\!+\!3)t,4t,3t)\geq $
	\begin{center}
		$\left\{\begin{array}{lll}
			\hskip-0.2cm	q^{6n}\!+\!\sum\limits_{i=1}^{n\!-\!4}\sum\limits_{j=i\!+\!1}^{n}q^{2(3n\!-\!i\!-\!1\!-\!2j)}\!+\!q^{6n\!-\!18}\!+\!3q^{6n\!-\!19}\!+\!3q^{6n\!-\!20}\!+\!q^{6n\!-\!21}\!+\!q^{12}\!+\!q^{8}\!+\!2q^{4}\!+\!q^{2}\!+\!1,&\!\!\!\mbox{if}  \  t=2, q>2;\\
			\hskip-0.2cm	q^{6n}\!+\!\sum\limits_{i=1}^{n\!-\!4}\sum\limits_{j=i\!+\!1}^{n}q^{2(3n\!-\!i\!-1-\!2j)}\!+\!q^{6n\!-\!18}\!+\!3q^{6n\!-\!19}\!+\!q^{12}\!+\!q^{8}\!+\!2q^{4}\!+\!q^{2}\!+\!1,&\!\!\!\mbox{if}  \ t=2, q=2;\\
			\hskip-0.2cm	q^{nt(t\!+\!1)}\!+\!\sum\limits_{i=1}^{n\!-\!4}\sum\limits_{j=i\!+\!1}^{n}q^{(n\!-\!i-1)t\!+\!(n\!-\!j)t^{2}}\!+\!q^{2t^{2}\!+\!2t}\!+\!q^{t^{2}\!+\!2t}\!+\!2q^{2t}\!+\!q^{t}\!+\!N\!+\!1,&\!\!\!\!\!\mbox{ if} \   t \neq 2;\\
		\end{array}\right.$
	\end{center}
	where  $N=q^{(t^{2}+t)n-\frac{13t^{2}}{4}-\frac{5t}{2}}+q^{(t^{2}+t)n-\frac{13t^{2}}{4}-3t}+2q^{(t^{2}+t)n-\frac{7t^{2}}{2}-\frac{5t}{2}}+2q^{(t^{2}+t)n-\frac{7t^{2}}{2}-3t}+q^{(t^{2}+t)n-\frac{15t^{2}}{4}-\frac{5t}{2}}+q^{(t^{2}+t)n-\frac{15t^{2}}{4}-3t}$.

	\item[$(2)$]If $n=5$, then   $\overline{A}_{q}(8t,4t,3t)
	\geq$
	\begin{center}
		$ \left\{\begin{array}{llll}
			q^{30}\!+\!q^{18}+q^{14}\!+\!q^{12}\!+\!7q^{10}\!+\!2q^{9}\!+\!q^{8}\!+q^6+\!2q^{4}\!+\!q^{2}\!+\!1,&\!\!\!\mbox{if} \  t=2, q>2;\\
			q^{30}\!+q^{18}+q^{14}\!+\!q^{12}\!+\!5q^{10}\!+\!q^{8}+q^6\!+\!2q^{4}\!+\!q^{2}\!+\!1,&\!\!\!\mbox{if} \  t=2, q=2;\\
			q^{5t(t\!+\!1)}\!+\!\sum\limits_{j=2}^{5}q^{3t\!+\!(5\!-\!j)t^{2}}\!+\!q^{t^{2}\!+\!3t}\!+\!q^{t^{2}\!+\!2t}\!+\!2q^{t^{2}\!+\!\frac{5t}{2}}\!+\!4q^{\frac{3t^{2}}{4}\!+\!\frac{5t}{2}}+q^{\frac{7t}{2}}+q^{3t}\!+\!2q^{2t}\!+\!q^{t}\!+\!1,&\!\!\! \mbox{if} \  t\neq2, q\geq 4;\\
			q^{5t(t\!+\!1)}\!+\!\sum\limits_{j=2}^{5}q^{3t+(5\!-\!j)t^{2}}\!+\!q^{t^{2}\!+\!3t}\!+\!2q^{t^{2}\!+\!\frac{5t}{2}}\!+\!4q^{\frac{3t^{2}}{4}\!+\!\frac{5t}{2}}\!+q^{\frac{7t}{2}}\!+\!2q^{3t}\!+\!2q^{2t}\!+\!q^{t}\!+\!1,&\!\!\!\! \ \mbox{if} \ t\neq2, q< 4.\\
		\end{array}\right.$
	\end{center}
\end{itemize}
\end{theo}

\begin{rem}
Similar to Theorems \ref{theo1} and  \ref{myth2}, one  can also  construct $((n+3)t,4t,3t)_q$ CDCs and calculate their dimensions when $t$ is odd.
\end{rem}

\subsection{Applications to improve the lower bound of $\overline{A}_q(n,d,k)$}
In this subsection, we apply Theorem \ref{myth2} to improve the lower bounds of $\overline{A}_q(n,8,6)$ for $16\leq n\leq 19$.

First,
as a direct consequence of Theorem \ref{myth2}, we get the following corollary.
\begin{coro}\label{coro1686}		
$$\overline{A}_{q}(16,8,6)
\geq \left\{\begin{array}{ll}
	q^{30}+q^{18}+q^{14}+q^{12}+7q^{10}+2q^{9}+q^{8}+q^{6}+2q^{4}+q^{2}+1,&\!\!\!\!\!\mbox{ if} \   q>2;\\
	q^{30}+q^{18}+q^{14}+q^{12}+5q^{10}+q^{8}+q^{6}+2q^{4}+q^{2}+1,&\!\!\!\mbox{if} \  q=2.
\end{array}\right.$$
\end{coro}

\begin{proof}
It follows from Theorem \ref{myth2} by taking $n=5$ and $t=2$.
\end{proof}

Based on Theorem \ref{myth2}, we can further  improve the lower bound of $\overline{A}_{q}(18,8,6)$ by   adding two  new identifying vectors.

\begin{coro}\label{coro1886}
$\overline{A}_{q}(18,8,6)
\geq$
\begin{center}
	$ \left\{\begin{array}{ll}
		q^{36}+q^{24}+q^{20}+2q^{18}+3q^{17}+4q^{16}+q^{15}+q^{14}+2q^{12}+2q^{10}+2q^{8}+q^{7}+q^{6}+2q^{4}+q^{2}+1,&\!\!\!\!\!\mbox{ if} \   q>2;\\
		q^{36}+q^{24}+q^{20}+2q^{18}+3q^{17}+q^{16}+q^{14}+2q^{12}+2q^{10}+2q^{8}+q^{7}+q^{6}+2q^{4}+q^{2}+1,&\!\!\!\mbox{if} \  q=2.\\
		
	\end{array}\right.$
\end{center}

\end{coro}

\begin{proof}
By taking $n=6$, $t=2$ in Theorem \ref{myth2}, we get that  $\overline{A}_{q}(18,8,6)
\geq$
\begin{center}
	$ \left\{\begin{array}{ll}
		q^{36}+q^{24}+q^{20}+2q^{18}+3q^{17}+4q^{16}+q^{15}+q^{14}+2q^{12}+q^{10}+2q^{8}+q^{6}+2q^{4}+q^{2}+1,&\!\!\!\!\!\mbox{ if} \   q>2;\\
		q^{36}+q^{24}+q^{20}+2q^{18}+3q^{17}+q^{16}+q^{14}+2q^{12}+q^{10}+2q^{8}+q^{6}+2q^{4}+q^{2}+1,&\!\!\!\mbox{if} \  q=2.\\
		
	\end{array}\right.$
\end{center}

Now let
$\nu_{1}=(000011101000001010)$ and $\nu_{2}=(000011010100000101)$.
Then $\mathcal{F}_{\nu_1}=[3, 4, 4, 4, 4, 4, 5, 6]$ and $\mathcal{F}_{\nu_2}=[2, 3, 4, 4, 4, 4, 4, 5 ]$. Considering the transposed Ferrers diagrams of $\mathcal{F}_{\nu_1}$  and $\mathcal{F}_{\nu_2}$, we get from Lemma  \ref{firstcons}
that there exists an optimal $[\mathcal{F}_{\nu_1},10,4]_q$ FDRMC and an optimal  $[\mathcal{F}_{\nu_2},7,4]_q$ FDRMC.  Lifting these two FDRMCs to CDCs, we get
a $(18,q^{10},8,6)_q$ CDC $\mathscr{C}_{\nu_1}$ and  a $(18,q^7,8,6)_q$ CDC $\mathscr{C}_{\nu_2}$.
Note that $d_H(\nu_1,\nu_2)=8$ and $d_H(\nu_i,v)\geq 8$ for $i\in \{1,2\}$ and each   identifying vector $v$ given  in Theorem \ref{theo1}. It follows from Lemma  \ref{mrlem2.10} that
$\mathscr D\cup \mathscr C_{\nu_1}\cup\mathscr C_{\nu_2}$
is a $(18,8,6)_q$ CDC. This completes the proof of the corollary.
\end{proof}

Next, we improve the  lower  bound of $\overline{A}_{q}(17,8,6)$ and $\overline{A}_{q}(19,8,6)$  by making a slight modification of the construction in Theorem \ref{theo1}.

To  construct a $(17,8,6)_q$ CDC,
let $X$ be the set of  all the identifying vectors used  in  the construction of $(16,8,6)_q$ CDC  in Corollary \ref{coro1686}. Set
$Y=\{(v,0)\mid v\in X\}$. Then we get the 19 identifying vectors in Table 3.
\begin{center}
Table 3

\medskip
\begin{tabular}{|c|cc|cc|}
	\hline
	${\rm set}$ &   ${\rm  vector}$ &  ${\rm size }$ & ${\rm   vector}$ & ${\rm size}$  \\
	\hline
	
	$\mathcal{A}$  & $\nu_{1}^{1}=\left( 11111100000000000\right) $&$q^{33} $ & & \\[-0.15em]
	\hline		\multirow{5}{*}{$\mathcal{B}$} &  $\nu_{1}^{2}=\left(11000011110000000 \right) $ &$q^{21} $ & $\nu_{2}^{2}=\left(11000011001100000 \right) $ &$q^{17} $  \\  [-0.15em]

	&
	$\nu_{3}^{2}=\left(11000011000011000 \right) $&$q^{13} $&$\nu_{4}^{2}=\left(11000011000000110 \right) $ &$q^{9} $ \\[-0.15em]
	
	&$\nu_{5}^{2}=\left(11000000111100000 \right) $&$q^{14} $ &
	$\nu_{6}^{2}=\left(11000000110011000 \right) $&$q^{10} $ \\[-0.15em]
	
	&$\nu_{7}^{2}=\left(11000000110000110 \right) $&$q^{6} $ &
	$\nu_{8}^{2}=\left(11000000001111000 \right) $&$q^{6} $ \\[-0.15em]
	
	&
	$\nu_{9}^{2}=\left(11000000001100110 \right) $&$q^{4} $ &$\nu_{10}^{2}=\left(11000000000011110 \right) $&$1 $
	\\
	\hline
	\multirow{2}{*}{$\mathcal{C}$} &
	$\nu_{1}^{3}=\left( 00110010101010000\right) $ &$q^{15} $ &
	$\nu_{2}^{3}=\left( 00110001011010000\right) $ &$q^{14} $  \\[-0.15em]
	
	&
	$\nu_{3}^{3}=\left( 00110001101001000\right)  $&$q^{14} $ &$\nu_{4}^{3}=\left( 00110001100110000\right)$&$q^{14} $\\ [-0.15em]
	&$\nu_{5}^{3}=\left( 00110010011001000\right) $ &$q^{13}$ &$\nu_{6}^{3}=\left( 00110010100101000\right)  $ &$q^{13}$
	\\[-0.15em]
	&$\nu_{7}^{3}=\left( 00110010010110000\right)$&$q^{13} $ &$\nu_{8}^{3}=\left( 00110001010101000\right) $ &$q^{12}$
	\\
	\hline
\end{tabular}
\end{center}

Now we calculate the sizes of FDRMCs corresponding to the identifying vectors given  in Table 3.  Denote by $\mathcal {F}^{2}_{j}$ the subdiagram of the  Ferrers diagram of $v^{2}_{j}$ by deleting the leftmost $4$ columns. Then we get that $\mathcal{F}^{2}_{5}=[2,2,6,6,6,6,6]$ and $\mathcal{F}^{2}_{9}=[2,2,2,2,4,4,6]$. \begin{itemize}
\item[$\bullet$]Let $\mathcal{F}'=[2,6,6,6,6,6]$ be the subdiagram of $\mathcal{F}_{5}^2$ by deleting the leftmost column.  Then $v_{\min}(\mathcal{F}',4)=14.$ It follows from Lemma \ref{firstcons} that there exists  an optimal $[\mathcal{F}',14,4]_q$  FDRMC. Then there exists an  $[\mathcal{F}^{2}_{5},14,4]_q$  FDRMC.  The  FDRMC corresponding to  $\nu^{2}_{6}$, $\nu^{2}_{7}$ and $\nu^{2}_{8}$ can be constructed similarly.
\item[$\bullet$]Let
$\mathcal{F}^{2}_{9}= \left(\begin{smallmatrix}
	\mathcal{F}_{1}&\mathcal{F}_{3}\\
	&\mathcal{F}_{2}
\end{smallmatrix}\right),$
where $\mathcal{F}_{1}=\begin{tikzpicture}[scale=0.3]

	\draw[fill=black] (0,0) circle(0.1);
	\draw[fill=black] (0.6,0) circle(0.1);
	\draw[fill=black] (1.2,0) circle(0.1);
	\draw[fill=black] (1.8,0) circle(0.1);
	\draw[fill=black] (0,-0.6) circle(0.1);
	\draw[fill=black] (0.6,-0.6) circle(0.1);
	\draw[fill=black] (1.2,-0.6) circle(0.1);
	\draw[fill=black] (1.8,-0.6) circle(0.1);
	
\end{tikzpicture}$ ,
$\mathcal{F}_{2}=\begin{tikzpicture}[scale=0.3]

	\draw[fill=black] (0,-0.6) circle(0.1);
	\draw[fill=black] (0.6,-0.6) circle(0.1);
	\draw[fill=black] (1.2,-0.6) circle(0.1);
	
	\draw[fill=black] (0,-1.0) circle(0.1);
	\draw[fill=black] (0.6,-1.0) circle(0.1);
	\draw[fill=black] (1.2,-1.0) circle(0.1);
	\draw[fill=black] (1.2,-1.4) circle(0.1);
	\draw[fill=black] (1.2,-1.8) circle(0.1);
\end{tikzpicture}$  and
$\mathcal{F}_{3}=\begin{tikzpicture}[scale=0.3]

	\draw[fill=black] (0,0) circle(0.1);
	\draw[fill=black] (0.6,0) circle(0.1);
	\draw[fill=black] (1.2,0) circle(0.1);
	
	\draw[fill=black] (0,-0.6) circle(0.1);
	\draw[fill=black] (0.6,-0.6) circle(0.1);
	\draw[fill=black] (1.2,-0.6) circle(0.1);

\end{tikzpicture}$ . By Lemma \ref{firstcons}, there exists an optimal $[\mathcal F^{t}_{1},4,2]_{q}$ FDRMC and an optimal $[\mathcal F_{2},4,2]_{q} $ FDRMC. It follows from Lemma \ref{comcons} that there exists an optimal $[\mathcal F^{2}_{9},4,4]_{q}$ FDRMC.
\end{itemize}	

It is easy to check by	Lemma \ref{firstcons}  that there exists an  optimal FDRMC for any one  of  the remaining  identifying vectors.
Consequently, we get the following result.

\begin{coro}\label{coro1786}$\overline{A}_{q}(17,8,6)
\geq$
\begin{center}
	$ \left\{\begin{array}{ll}
		q^{33}+q^{21}+q^{17}+q^{15}+4q^{14}+4q^{13}+q^{12}+q^{10}+q^{9}+2q^{6}+q^{4}+1,&\!\!\!\!\!\mbox{ if} \   q>2;\\
		q^{33}+q^{21}+q^{17}+q^{15}+4q^{14}+q^{13}+q^{10}+q^{9}+2q^{6}+q^{4}+1,&\!\!\!\mbox{if} \  q=2.\\
		
	\end{array}\right.$
\end{center}
\end{coro}

Next, we construct a $(19,8,6)_q$ CDC.  Let $X$ be the set of  all identifying vectors used  in  the construction of $(18,8,6)_q$ CDC  in Corollary \ref{coro1886}.  Set
$Y\!=\!\{(v,0)\!\mid\! v\in X\}$.  All the Ferrers diagrams which correspond to the identifying vectors from $\mathcal{B}$ have a common quasi-pending block in the the 3 leftmost. We fill the same matrix of $ [2\times 3,q^{3},2]_{q}$ MRD code $\mathcal{M}$ in the quasi-pending block for each vector in the same class. For the identifying vectors from different classes, we fill different matrices of $\mathcal{M}$. Obviously, $q^{3}\geq 6$ for any prime power $q$. Then we get  the 26 identifying vectors in Table 4.
\begin{center}
Table 4

\medskip
\begin{tabular}{|c|cc|cc|}
	
	\hline
	${\rm set}$ &   ${\rm  vector}$ &  ${\rm size}$ & ${\rm  vector}$ & ${\rm size}$  \\ \hline
	
	$\mathcal{A}$  & $\nu_{1}^{1}=\left( 1111110000000000000\right) $ &$q^{39} $ & &  \\\hline
	\multirow{8}{*}{$\mathcal{B}$} &  $\nu_{1}^{2}=\left(1100001111000000000 \right) $ &$q^{27} $ & $\nu_{2}^{2}=\left(1100001100110000000 \right) $ &$q^{23} $   \\ [-0.15em]

	&$\nu_{3}^{2}=\left(1100001100001100000 \right) $&$q^{19} $&$\nu_{4}^{2}=\left(1100001100000011000 \right) $ &
	$q^{15} $ \\ [-0.15em]
	
	&$\nu_{5}^{2}=\left(1100001100000000110 \right) $ &
	$q^{11} $ &
	$\nu_{6}^{2}=\left(1100000011110000000 \right) $
	&$q^{21} $  \\ [-0.15em]
	
	&$\nu_{7}^{2}=\left(1100000011001100000\right)$&$q^{17} $ &
	$\nu_{8}^{2}=\left(1100000011000011000 \right) $ &$q^{13} $ \\ [-0.15em]
	
	&
	$\nu_{9}^{2}=\left(1100000011000000110 \right) $&$q^{9} $  &$\nu_{10}^{2}=\left(1100000000111100000 \right) $&$q^{15} $\\ [-0.15em]
	&
	$\nu_{11}^{2}=\left(1100000000110011000 \right) $ &$q^{11} $ &
	$\nu_{12}^{2}=\left(1100000000110000110 \right) $ &$q^{7} $ \\ [-0.15em]
	
	&$\nu_{13}^{2}=\left(1100000000001111000 \right) $&$q^{7} $ &
	$\nu_{14}^{2}=\left(1100000000001100110 \right) $ &$q^{4} $ \\ [-0.15em]
	&
	$\nu_{15}^{2}=\left(1100000000000011110 \right) $ &$1 $&&\\
	\hline
	\multirow{5}{*}{$\mathcal{C}$} &
	$\nu_{1}^{3}=\left( 0011001010101000000\right) $ &$q^{21} $ &
	$\nu_{2}^{3}=\left( 0011000101101000000\right) $ &$q^{20} $  \\ [-0.15em]
	
	&
	$\nu_{3}^{3}=\left( 0011000110100100000\right)  $&$q^{20} $ &$\nu_{4}^{3}=\left( 0011000110011000000\right)$&$q^{20} $ \\ [-0.15em]
	&$\nu_{5}^{3}=\left( 0011001001100100000\right) $ &$q^{19} $  &$\nu_{6}^{3}=\left( 0011001010010100000\right) $ &$q^{19} $
	\\ [-0.15em]
	
	&$\nu_{7}^{3}=\left( 0011001001011000000\right) $&$q^{19} $ &$\nu_{8}^{3}=\left( 0011000101010100000\right) $ &$q^{18} $
	\\ [-0.15em]
	
	&$\nu_{9}^{3}=(0000111010000010100)$ &$q^{13} $&$\nu_{10}^{3}=(0000110101000001010)$&$q^{10} $\\
	\hline
\end{tabular}
\end{center}
\medskip

Now we calculate the sizes of FDRMCs corresponding to the identifying vectors given  in Table 4.
Denote by $\mathcal {F}^{2}_{j}$ the subdiagram of the  Ferrers diagram of $v^{2}_{j}$ by deleting the leftmost $3$ columns. 
\begin{itemize}
\item[$\bullet$]Since $\mathcal{F}^{2}_{10}=[2,2,2,2,2,6,6,6,6,6]$, we know that $v_{min}(\mathcal{F}^{2}_{10} , 4)=15$. Considering the transposed Ferrers diagram of $\mathcal{F}^{2}_{10}$, we get from Lemma \ref{gecons} that there exists an optimal $[\mathcal{F}^{2}_{10},15,4]_{q}$ FDRMC. The  FDRMC corresponding to  $\nu^{2}_{11}$ and $\nu^{2}_{12}$ can be constructed similarly.

\item[$\bullet$]Let
$\mathcal{F}^{2}_{13}= \left(\begin{smallmatrix}
	\mathcal{F}_{1}&\mathcal{F}_{3}\\
	&\mathcal{F}_{2}
\end{smallmatrix}\right),$
where $\mathcal{F}_{1}=\begin{tikzpicture}[scale=0.3]

	\draw[fill=black] (0,0) circle(0.1);
	\draw[fill=black] (0.6,0) circle(0.1);
	\draw[fill=black] (1.2,0) circle(0.1);
	\draw[fill=black] (1.8,0) circle(0.1);
	\draw[fill=black] (0,-0.6) circle(0.1);
	\draw[fill=black] (0.6,-0.6) circle(0.1);
	\draw[fill=black] (1.2,-0.6) circle(0.1);
	\draw[fill=black] (1.8,-0.6) circle(0.1);
	\draw[fill=black] (2.4,0) circle(0.1);
	\draw[fill=black] (3,0) circle(0.1);
	\draw[fill=black] (2.4,-0.6) circle(0.1);
	\draw[fill=black] (3,-0.6) circle(0.1);
	\draw[fill=black] (3.6,0) circle(0.1);
	\draw[fill=black] (3.6,-0.6) circle(0.1);
\end{tikzpicture}$ ,
$\mathcal{F}_{2}=\begin{tikzpicture}[scale=0.3]

	\draw[fill=black] (0,0) circle(0.1);
	\draw[fill=black] (0.6,0) circle(0.1);
	\draw[fill=black] (1.2,0) circle(0.1);
	
	\draw[fill=black] (0,-0.5) circle(0.1);
	\draw[fill=black] (0.6,-0.5) circle(0.1);
	\draw[fill=black] (1.2,-0.5) circle(0.1);
	\draw[fill=black] (1.2,-1.0) circle(0.1);
	\draw[fill=black] (1.2,-1.5) circle(0.1);
	\draw[fill=black] (0.6,-1.0) circle(0.1);
	\draw[fill=black] (0.6,-1.5) circle(0.1);
	\draw[fill=black] (0,-1.0) circle(0.1);
	\draw[fill=black] (0,-1.5) circle(0.1);
\end{tikzpicture}$  and
$\mathcal{F}_{3}=\begin{tikzpicture}[scale=0.3]

	\draw[fill=black] (0,0) circle(0.1);
	\draw[fill=black] (0.6,0) circle(0.1);
	\draw[fill=black] (1.2,0) circle(0.1);
	
	\draw[fill=black] (0,-0.6) circle(0.1);
	\draw[fill=black] (0.6,-0.6) circle(0.1);
	\draw[fill=black] (1.2,-0.6) circle(0.1);

\end{tikzpicture}$ . By Lemma \ref{firstcons}, there exists an optimal $[\mathcal F^{t}_{1},7,2]_{q}$ FDRMC and an optimal $[\mathcal F_{2},8,2]_{q} $ FDRMC. It follows from Lemma \ref{comcons} that there exists an optimal $[\mathcal F^{2}_{14},7,4]_{q}$ FDRMC.

\item[$\bullet$]Let
$\mathcal{F}^{2}_{14}= \left(\begin{smallmatrix}
	\mathcal{F}'_{1}&\mathcal{F}'_{3}\\
	&\mathcal{F}'_{2}
\end{smallmatrix}\right),$
where $\mathcal{F}'_{1}=\begin{tikzpicture}[scale=0.3]

	\draw[fill=black] (0,0) circle(0.1);
	\draw[fill=black] (0.6,0) circle(0.1);
	\draw[fill=black] (1.2,0) circle(0.1);
	\draw[fill=black] (1.8,0) circle(0.1);
	\draw[fill=black] (0,-0.6) circle(0.1);
	\draw[fill=black] (0.6,-0.6) circle(0.1);
	\draw[fill=black] (1.2,-0.6) circle(0.1);
	\draw[fill=black] (1.8,-0.6) circle(0.1);
	\draw[fill=black] (2.4,0) circle(0.1);
	\draw[fill=black] (3,0) circle(0.1);
	\draw[fill=black] (2.4,-0.6) circle(0.1);
	\draw[fill=black] (3,-0.6) circle(0.1);
	\draw[fill=black] (3.6,0) circle(0.1);
	\draw[fill=black] (3.6,-0.6) circle(0.1);
\end{tikzpicture}$ ,
$\mathcal{F}'_{2}=\begin{tikzpicture}[scale=0.3]

	\draw[fill=black] (0,0) circle(0.1);
	\draw[fill=black] (0.6,0) circle(0.1);
	\draw[fill=black] (1.2,0) circle(0.1);
	
	\draw[fill=black] (0,-0.4) circle(0.1);
	\draw[fill=black] (0.6,-0.4) circle(0.1);
	\draw[fill=black] (1.2,-0.4) circle(0.1);
	\draw[fill=black] (1.2,-1.0) circle(0.1);
	\draw[fill=black] (1.2,-1.6) circle(0.1);
\end{tikzpicture}$  and
$\mathcal{F}'_{3}=\begin{tikzpicture}[scale=0.3]

	\draw[fill=black] (0,0) circle(0.1);
	\draw[fill=black] (0.6,0) circle(0.1);
	\draw[fill=black] (1.2,0) circle(0.1);
	
	\draw[fill=black] (0,-0.6) circle(0.1);
	\draw[fill=black] (0.6,-0.6) circle(0.1);
	\draw[fill=black] (1.2,-0.6) circle(0.1);

\end{tikzpicture}$ . By Lemma \ref{firstcons}, there exists an optimal $[(\mathcal F'_{1})^{t},7,2]_{q}$ FDRMC and an optimal $[\mathcal F'_{2},4,2]_{q} $ FDRMC. It follows from Lemma \ref{comcons} that there exists an optimal $[\mathcal F^{2}_{14},4,4]_{q}$ FDRMC.

\end{itemize}	 	

By Lemma \ref{firstcons}, it is easy to  check  that there exists an  optimal FDRMC for any  one of the remaining  identifying vectors.
Consequently, we get the following corollary.
\begin{coro}
\label{coro1986}$\overline{A}_{q}(19,8,6)
\geq$
\begin{center}
		$ \left\{\begin{array}{ll}
		q^{39}+q^{27}+q^{23}+2q^{21}+3q^{20}+4q^{19}+q^{18}+q^{17}+2q^{15}+2q^{13}+2q^{11}+q^{10}+q^{9}+2q^{7}+q^{4}+1,&\!\!\!\!\!\mbox{ if} \   q>2;\\
		q^{39}+q^{27}+q^{23}+2q^{21}+3q^{20}+q^{19}+q^{17}+2q^{15}+2q^{13}+2q^{11}+q^{10}+q^{9}+2q^{7}+q^{4}+1,&\!\!\!\mbox{if} \  q=2.\\
		
	\end{array}\right.$
\end{center}
\end{coro}

\begin{rem} These lower bounds of $\overline{A}_q(n,k,d)$ provided in this section is greater than the previously  best known  bounds (see Table 5).
\end{rem}

\begin{center}
Table 5

\medskip
\begin{tabular}{|  p{2cm} |  p{5cm} |  p{8cm} |}
	
	\hline	
	${\rm \overline{A}_{q}(n,2\delta,k)}$  & {\rm Old lower bounds\cite{subspace}}  &{\rm New lower bounds}   \\
	\hline
	$\overline{A}_{2}(16,8,6)$  & 1074024641 & 1074029925 [Cor.\ref{coro1686}]\\
	\hline
	$\overline{A}_{3}(16,8,6)$  & 205891524417339 & 205891525289719 [Cor.\ref{coro1686}] \\
	\hline
	$\overline{A}_{4}(16,8,6)$  & 1152921573596861952 & 1152921573619470865 [Cor.\ref{coro1686}] \\
	\hline
	$\overline{A}_{2}(17,8,6)$  & 8592201220 &8592270993  [Cor.\ref{coro1786}]\\
	\hline
	$\overline{A}_{3}(17,8,6)$  &5559071026909483  &5559071196518677   [Cor.\ref{coro1786}]\\
	\hline
	$\overline{A}_{4}(17,8,6)$  & 73786980692884721729
	&73786980712498602241    [Cor.\ref{coro1786}]\\
	\hline
	$\overline{A}_{2}(18,8,6)$  & 68737628866 & 68738312933   [Cor.\ref{coro1886}]\\
	\hline
	$\overline{A}_{3}(18,8,6)$  & 150094917726595498 & 150094922568097420 [Cor.\ref{coro1886}] \\
	\hline
	$\overline{A}_{4}(18,8,6)$  & 4722366764344622977361 & 4722366765651669963281  [Cor.\ref{coro1886}]\\
	\hline
	$\overline{A}_{2}(19,8,6)$  &549901076262 &549906503441 [Cor.\ref{coro1986}]\\
	\hline
	$\overline{A}_{3}(19,8,6)$  &4052562778621323226  &4052562909338630152   [Cor.\ref{coro1986}] \\
	\hline
	$\overline{A}_{4}(19,8,6)$  &  302231472918056072905793 &302231473001706877649153   [Cor.\ref{coro1986}] \\\hline
	
	\hline
\end{tabular}
\end{center}

\section{Conclusions}
The paper is devoted to constructing CDCs by the multilevel construction.  We  choose  our    skeleton code based on the transformations of the vectors related to
a  one-factorization of a complete graph, and calculate the dimensions by use of known constructions of
optimal FDRMCs. As an applications, we improve the lower bounds of $\overline{A}_q(n,8,6)$
for $16\leq n\leq 19.$ It should be noted that there are several Ferrers diagrams  for which we can not construct optimal FDRMCs in our constructions. If one  can construct optimal FDRMCs for these diagrams, then the lower bound listed in Table 5 should  be further  improved.

\medskip
\noindent{\bf Acknowledgements} Dengming Xu  is partially supported by the Natural Science Foundation of Tianjin (Grant No.23JCQNJC00050)

\bibliographystyle{unsrt}
\bibliography{paper}

\end{document}